\documentclass[12pt]{article}
\usepackage{amsmath}
\usepackage{amsthm}
\usepackage{amssymb}
\newtheorem{theorem}{Theorem}
\newtheorem{conjecture}{Conjecture}
\newtheorem{lemma}{Lemma}
\DeclareMathOperator*{\E}{E}
\DeclareMathOperator*{\sspan}{span}
\newcommand{\bbF}{\mathbb{F}}
\newcommand{\bbN}{\mathbb{N}}
\newcommand{\<}{\langle}
\renewcommand{\>}{\rangle}

\newcommand{\poly}{\mathrm{poly}}

\usepackage{mathtools}
\title{The PFR Conjecture Holds for\\Two Opposing Special Cases}
\author{Thomas Holenstein}
\begin{document}
\maketitle
\begin{abstract}
Let $A \subseteq \bbF_2^n$ be a set with $|2A| = K|A|$.
We prove that if 
\begin{enumerate}
\item[(1)]
for at least a fraction $1-K^{-9}$ of all $s \in 2A$, 
the set $(A+s) \cap A$ has size
\emph{at most} $L\cdot|A|/K$, or
\item[(2)] for at least a fraction $K^{-L}$ of all $s \in 2A$,
the set $(A+s) \cap A$ has size \emph{at least} $|A|\cdot(1- K^{-1/L})$,
\end{enumerate}
then there is a subset $B \subseteq A$ of size $|A|/K^{O_L(1)}$ such 
that $\sspan(B) \leq K^{O_L(1)}\cdot|A|$.
\end{abstract}
\section{Introduction}
If $A$ is a subset of an abelian group $G$ (which will always 
be $\bbF_2^n$ in this paper), we let 
$A+A := 2A := \{a_1+a_2 | a_1,a_2 \in A\}$.
Suppose that $|A+A| = K|A|$ (where one should think of $K$ as 
``small'' compared to $|A|$).
Intuitively, one expects $A$ to have a lot of structure in this case.
The Freiman-Ruzsa Theorem (which we quote in a 
a version implied by \cite{EveZoh12}), gives such structure.
\begin{theorem}[Freiman-Ruzsa]
Let $A \subseteq \bbF_2^{n}$ be a set with $|2A| = K|A|$.
Then, $\sspan(A) \leq 2^{2K}|A|$.
\end{theorem}

The exponential loss in $K$ is necessary, as can be seen
by considering the set $A \subseteq \bbF_2^{n}$, which 
contains all $a$ with exactly one $1$ in the first $t$ positions for 
an appropriately chosen~$t$.
In fact, the theorem given in \cite{EveZoh12} describes $\sspan(A)/|A|$
correctly in dependence of $|2A|/|A|$.

One can conjecture that the exponential loss is unneccessary
if one is allowed to take a large subset of $A$.
This is called the Polynomial Freiman-Ruzsa Conjecture, and seems
to be due to Marton (see \cite{Ruzsa93}).
\begin{conjecture}[Polynomial Freiman-Ruzsa Conjecture]\label{con:pfr}
There exists $c \in \bbN$ such that, for any 
set $A \subseteq \bbF_2^{n}$ with $|2A| = K|A|$
there is a subset $A' \subseteq A$ such that 
$|A'| \geq K^{-c} |A|$ and $\sspan(A') \leq K^{c}|A|$.
\end{conjecture}
The strongest result in this direction was given by
Sanders \cite{Sanders12}.
\begin{theorem}[Sanders]\label{thm:sanders}
For any set $A \subseteq \bbF_2^{n}$ with $|2A| = K|A|$
there is a subset $A' \subseteq A$ such that 
$|A'| \geq 2^{-\log^{4}(K)} |A|$ and $\sspan(A') |A|$.
\end{theorem}
The proof of Theorem~\ref{thm:sanders} is based 
on a result by Crook and Sisask \cite{CroSis10}.

We remark that the validity of Conjecture~\ref{con:pfr} was
proven in the special case that $A$ is a downset \cite{GreTao09}.

\subsection{Our Contributions}
For a condition $C$, we use $\<C\>$ as Iverson bracket, i.e.,
$\<C\>$ is $1$ if $C$ is satisfied, and 0 otherwise.

For $s \in 2A$, let $A(s) := A \cap (s+A) = \{a \in A :
\exists a' \in A: s = a+a'\}$.
Because $\sum_{s\in 2A} |A(s)| =
\sum_{s \in 2A} \sum_{a_1 \in A} \<\exists a_2: a_1 + a_2 = s\>
=\sum_{s \in 2A} \sum_{a_1,a_2 \in A} \<a_1 + a_2 = s\>
=\sum_{a_1,a_2 \in A}1 = |A|^2$ we see that 
$\E_{s \in 2A}[A(s)] = |A|/K$.

We show that the PFR conjecture holds for two special 
cases.
In Theorem~\ref{thm:unstructured}, 
we show that in case most (namely, a fraction $1-K^{-9}$) of the
sets $A(s)$ have size close to the minimum they can have on average
(i.e., $|A|/K$), then the PFR conjecture holds.
In Theorem~\ref{thm:structured}, we prove that if a large
fraction of the sets $A(s)$ have size close to the maximum they can have
(namely $|A|$), then the PFR conjecture holds.

For both theorems it is easy to find examples of sets $A$ which
satisfy the conditions of the theorem.
\subsection{Preliminaries}

We use the Balog-Szemer\'edi-Gowers Theorem, whose formulation 
we take from \cite{BeLoZe12}.
\begin{theorem}[Balog-Szemer\'edi-Gowers]\label{thm:bsg}
There exist polynomials $f(x,y)$, $g(x,y)$ such that the following
holds for every subset $A$ of an abelian group $G$.
If $A$ satisfies $\Pr_{a_1,a_2 \in A}[a_1 + a_2 \in S] \geq \frac{1}{K}$
for some $S\subseteq G$, $|S| \leq C|A|$, then there is $A' \subseteq A$
such that $|A'| \geq |A|/f(K,C)$ and $|A'+A'| \leq g(K,C)|A|$.
\end{theorem}

\section{The Unstructured Case}
\begin{theorem}\label{thm:unstructured}
Let $A \subseteq \bbF_2^n$, $|A+A| = K|A|$.
If 
\begin{align}\label{eq:2}
\Pr_{a_1,a_2\in A}
\Bigl[|A(a_1+a_2)| \geq \frac{L|A|}{K}\Bigr] \leq \frac{1}{K^{8}}
\end{align}
then there exists $B \subseteq A$ with $|B| \geq |A|/\poly(L)$ 
such that $\sspan(B) \leq |A+A|2^{\poly(L)}$.
\end{theorem}
We note that statement (1) in the abstract is implied by 
Theorem~\ref{thm:unstructured}.
To see this, we observe that
$\Pr_{a_1,a_2\in A}
\bigl[|A(a_1+a_2)| \geq \frac{L|A|}{K}\bigr]
\leq
K\cdot \Pr_{s \in 2A}
\bigl[|A(s)| \geq \frac{L|A|}{K}\bigr]$
because $\Pr_{a_1,a_2 \in A}[a_1 + a_2 = s] \leq \frac{1}{|A|} \leq 
K \Pr_{s' \in 2A}[s = s']$ for any $s$.

\medskip

The idea of the proof of Theorem~\ref{thm:unstructured}
is to pick values $\beta_1,\beta_2,\beta_3,\beta_4$ uniformly at random 
from $A$, and show that (\ref{eq:2}) 
implies that with probability depending only on $L$ we have
\begin{align}\label{eq:3}
A(\beta_1+\beta_2) \cap A(\beta_3+\beta_4) \neq \emptyset.
\end{align}
Suppose now that indeed~(\ref{eq:3}) holds, and let $a$ be in the intersection.
Then, because $\beta_1+\beta_2 = a+a_1$ and $\beta_3+\beta_4 = a+a_2$, 
we see that $\beta_1+\beta_2+\beta_3+\beta_4 = a_1 + a_2 \in 2A$.
Thus, in this case we get that
$\Pr_{s_1,s_2 \in 2A}[s_1+s_2 \in 2A]$ only depends on $L$
(actually this is not quite true, as the distribution is somewhat wrong,
but we ignore this for this informal discussion).
Thus, Theorem~\ref{thm:bsg} gives us a subset of $2A$ whose doubling constant 
only depends on~$L$. 
By the Freiman-Ruzsa Theorem, the span of this subset is as small as required,
and with somewhat more work we can
get a large subset of $A$ which also has small span.

In the remainder of this section, we show that (\ref{eq:3}) happens 
with probability dependent only on $L$.
After this, we make the above intuition formal to get the result.

To prove that~(\ref{eq:3}) happens with probability dependent only
on $L$, we define the random variable
\begin{align}
Y &:= |A(\beta_1+\beta_2) \cap A(\beta_3+\beta_4)|\nonumber\\
&\phantom{:}=
\sum_{a \in A} \<a \in A(\beta_1+\beta_2)\> \<a \in A(\beta_3+\beta_4)\>.
\end{align}
and use $\Pr[Y>0] \geq \frac{\E^2[Y]}{\E^2[Y]}$.

For technical reasons (the ``somewhat more work'' above), 
we later need actually that 
not only (\ref{eq:3}) happens, but at the same time
the sets $A(\beta_1+\beta_2)$ and $A(\beta_3+\beta_4)$ are not too small.
Because of this, we actually work with the the random variable $Z$:
\begin{align}
Z := \begin{cases}
0 & \text{if $|A({\beta_1+\beta_2})| \leq \frac{|A|}{2K}$ or 
$|A({\beta_3+\beta_4})| \leq \frac{|A|}{2K}$}\\
Y & \text{otherwise}.
\end{cases}
\end{align}
We will show that~(\ref{eq:2}) implies $\Pr[Z>0] \geq \Omega(L^{-4})$.
For this, we first show that a simple application of Cauchy-Schwarz 
gives a good lower bound on $\E[Z]$.

\begin{lemma}\label{lem:lowerBoundOnZ}
If $A \subseteq \bbF_2^n$, $|A+A| = K|A|$, 
then $\E[Z] \geq \frac{|A|}{16 K^2}$
(where the expectation is
over the uniform random choice of $\beta_1$ to $\beta_4$ in $A$).
\end{lemma}
\begin{proof}
We see that
\begin{align}
\E[Z] &= 
\E\Bigl[ \sum_{a \in A}
\<a \in A({\beta_1+\beta_2}) \land |A({\beta_1+\beta_2})|\geq\tfrac{|A|}{2K} \>
\times\nonumber\\
&\qquad\qquad\qquad\qquad
\<a \in A({\beta_3+\beta_4}) \land |A({\beta_3+\beta_4})|\geq\tfrac{|A|}{2K}\>\Bigr]
\\
& =
 \sum_{a \in A}
\Pr\Bigl[a \in A({\beta_1+\beta_2}) \land |A({\beta_1+\beta_2})|\geq\frac{|A|}{2K}\Bigr]^2\\
&\geq
\frac{
\Bigl(\sum_{a \in A}
\Pr[a \in A({\beta_1+\beta_2}) \land |A({\beta_1+\beta_2})|\geq\frac{|A|}{2K}]\Bigr)^2}{|A|}\label{eq:1}
\end{align}

We now observe that
\begin{align}
\Pr_{\beta_1,\beta_2 \in A}\Bigl[|A({\beta_1+\beta_2})| \leq \frac{|A|}{2K}\Bigr]
\leq
\frac{|A+A| \frac{|A|}{2K}}{|A|^2} = \frac12.
\end{align}
and so
\begin{align}
\MoveEqLeft{\sum_{a \in A} 
\Pr\Bigl[a \in A({\beta_1+\beta_2}) \land |A({\beta_1+\beta_2})| 
\geq \tfrac{|A|}{2K}
\Bigr]}\\
&\geq\frac{1}{|A|^2}
\sum_{\beta_1,\beta_2: |A({\beta_1+\beta_2})| \geq |A|/2K} 
\sum_{a}  \<a \in A({\beta_1+\beta_2})\> \\
&\geq\frac{1}{|A|^2}
\sum_{\beta_1,\beta_2: |A({\beta_1+\beta_2})| \geq |A|/2K} 
\frac{|A|}{2K} 
\geq \frac{|A|}{4K}
\end{align}
Inserting into (\ref{eq:1}) gives the lemma.
\end{proof}

We next prove that $\E[Z^2]$ is small; of course, it suffices to show
that $\E[Y^2]$ is small.

\begin{lemma}\label{lem:YsquareIsSmall}
Let $A \subseteq \bbF_2^n$, $|A+A| = K|A|$.
Suppose that \begin{align}
\Pr_{a_1,a_2 \in A}\Bigl[|A(a_1+a_2)| \geq \frac{L|A|}{K}\Bigr] \leq 
\frac{1}{K^{8}}.
\end{align}
Then,
$E[Y^2] \leq 6\frac{L^4|A|^2}{K^4}$.
\end{lemma}

We sketch the proof of Lemma~\ref{lem:YsquareIsSmall}: 
linearity of expectation gives
\begin{align}
\E[Y^2] &= \E\Bigl[\sum_{a_1,a_2 \in A} \<\{a_1,a_2\} \subseteq A({\beta_1+\beta_2})\>
\<\{a_1,a_2\} \subseteq A({\beta_3+\beta_4})\>\Bigr]\nonumber\\
&=
\sum_{a_1,a_2 \in A} \Pr[\{a_1,a_2\} \subseteq A({\beta_1+\beta_2})]^2 \label{eq:4}
\end{align}
In order to give an upper bound on the probabilities in this expression,
the following lemma is key.
\begin{lemma}\label{lem:keyStep}
For any $A \subseteq \bbF_2^n$ and any $a_1,a_2 \in A$ we have
\begin{align}
&\bigl|\{(\beta_1,\beta_2) \in A\times A : \{a_1,a_2\} \subseteq A({\beta_1+\beta_2})\}\bigr|\nonumber\\
&\qquad =
\bigl|\{(c_1,\beta_2) \in A \times A : c_1 \in A({a_1+a_2}) \land \beta_2 \in A({a_1+c_1})\bigr|\;.\label{eq:222}
\end{align}
\end{lemma}
\begin{proof}
We show that the map 
$(\beta_1,\beta_2) \mapsto (a_1+\beta_1+\beta_2,\beta_2)$ is a bijection 
of the two sets.
Clearly, the map is bijective as a map in $\bbF_2^{*}$.

Now, suppose that $(\beta_1,\beta_2)$ is in the first set.
Let $c_1 := a_1+\beta_1+\beta_2$, and also set $c_2 := a_2+\beta_1+\beta_2$.
We see that both $c_1$ and $c_2$ are in $A$ (because $a_1 \in A({\beta_1+\beta_2})$, and
$a_2 \in A({\beta_1+\beta_2})$).
Since also $c_1+c_2 = a_1+a_2$ we see that $c_1 \in A({a_1+a_2})$.
Furthermore, because $\beta_1+\beta_2 = a_1+c_1$ we see that 
$\beta_2 \in A({a_1+c_1})$.
Thus, $(a_1+\beta_1+\beta_2,\beta_2)$ is in the second set.

Next, suppose that $(c_1,\beta_2)$ is in the second set.
Let $\beta_1 := a_1+\beta_2+c_1$, so that the preimage of
$(c_1,\beta_2)$ is $(\beta_1,\beta_2)$.
We first notice that $\beta_1 = a_1+\beta_2+c_1 \in A$ 
because $\beta_2 \in A({a_1+c_1})$.
Next, we clearly have $a_1 \in A({\beta_1+\beta_2}) = A({a_1+c_1})$.
Finally, we also have $a_2 \in A({a_1+c_1})$: this is
the same condition as $c_1 \in A({a_1+a_2})$.
\end{proof}
Continuing our proof sketch of Lemma~\ref{lem:YsquareIsSmall}, 
we suppose for a moment that $|A(s)| \leq L|A|/K$ 
for all $s \in 2A$.
Then, the set on the right hand side in (\ref{eq:222}) is clearly at 
most of cardinality $L^2|A|^2/K^2$, and by (\ref{eq:4})
we see that $\E[Y^2] \leq L^4|A|^2/K^4$, as we want to show.

Thus, all which remains to do is to show that in case $|A(s)| \leq L|A|/K$
with probability $1-K^{-8}$ we can still give the bound; we have chosen
the $8$ in the exponent to make this possible.\footnote{We remark that
the assumption 
$\forall s: |A(s)| \leq L|A|/K$ does not really make sense, because 
$A(0) = A$.}
\begin{proof}[Proof of Lemma~\ref{lem:YsquareIsSmall}]
We start with
\begin{align}\label{eq:5}
\E[Y^2] = \sum_{a_1,a_2} \Pr[\{a_1,a_2\} \subseteq A({\beta_1+\beta_2})]^2
\end{align}
and
\begin{align}
\Pr[\{a_1,a_2\} \subseteq A({\beta_1+\beta_2})]
=
\frac{1}{|A|^2} \sum_{\beta_1,\beta_2 \in A} \<\{a_1,a_2\} \subseteq A({\beta_1+\beta_2})\}\>\;.
\end{align}
By Lemma~\ref{lem:keyStep}, 
\begin{align}
\Pr[\{a_1,a_2\} \subseteq A({\beta_1+\beta_2})]
=\frac{1}{|A|^2} \sum_{c_1 \in A({a_1+a_2})}|A({a_1+c_1})|\;.
\end{align}

For $a \in A$ let $B(a) := \{ b \in A : |A({a+b})| \geq L|A|/K\}$.
Then, we let $A'$ be the set of elements for which $|B(a)|$ is large.
\begin{align}
A'= \{ a :  |B(a)| \geq |A|/K^4 \}
\end{align}
Note that $|A'| \leq |A|/K^4$, as otherwise
$\Pr_{a_1,a_2}[|A({a_1+a_2})| \geq L|A|/K] \geq \frac{1}{K^{8}}$.

To now upper bound the right hand side of (\ref{eq:5}),
we first fix $a_1 \in A \setminus A'$
and $a_2 \in A \setminus B(a_1)$.
Then, we have
\begin{align}
\lefteqn{|A|^4 \cdot \Pr[\{a_1,a_2\} \subseteq A({\beta_1+\beta_2})]^2}\\
&=
\Bigl(
\sum_{c_1 \in A({a_1+a_2}) \cap B(a_1)} |A({a_1+c_1})|
+ \sum_{c_1 \in A({a_1+a_2}) \setminus B(a_1)} |A({a_1+c_1})|
\Bigr)^2\\
&\leq
\Bigl(
\sum_{c_1 \in B(a_1)} |A|
+ \sum_{c_1 \in A({a_1+a_2})} \frac{L|A|}{K}
\Bigr)^2\\
&\leq
\Bigl(
|B(a_1)| |A|
+ |A({a_1+a_{2}})| \frac{L|A|}{K}
\Bigr)^2\\
&\leq
\Bigl(
\frac{|A|^2}{K^4}
+  \frac{L^2|A|^2}{K^2}
\Bigr)^2 \leq \frac{4L^4|A|^4}{K^4}\;.
\end{align}

Thus, 
\begin{align}
\sum_{a_1,a_2 \in A} \Pr[&\{a_1,a_2\} \subseteq A({\beta_1+\beta_2})]^2\\
&= 
\sum_{a_1 \notin A', a_2 \notin B(a_1)} 4\frac{L^{4}}{K^4}
+
\sum_{a_1 \notin A', a_2 \in B(a_1)} 1
+
\sum_{a_1 \notin A', a_2 \in A} 1\\
&\leq
4|A|^2\frac{L^4}{K^4} + |A|^2 \frac{1}{K^4} + 
|A|^2 \frac{1}{K^4}
\\
&\leq
6|A|^2\frac{L^4}{K^4}\;.
\end{align}
\end{proof}

We can now prove Theorem~\ref{thm:unstructured}.
\begin{proof}[Proof of Theorem~\ref{thm:unstructured}]
Let $C \subseteq 2A$ be the elements $s$ of $2A$
for which $A(s)$ has ``typical'' size, 
formally $C := \{s :  \frac{|A|}{2K} \leq |A(s)| \leq L \frac{|A|}{K}\}$.
The total number of pairs 
$(a_1,a_2) \in A\times A$ is at most
$\frac{|2A|}{K^8} \cdot |A| + 
|C|\cdot L\frac{|A|}{K} + |2A|\cdot \frac{|A|}{2K}
=
|A|^2 \Bigl(\frac{1}{K^7} + L\frac{|C|}{|2A|} + \frac{1}{2}\Bigr)$.
Because the number of pairs equals $|A|^2$, and we can
assume $K \geq 2$ (otherwise the theorem is easily seen to hold),
we have $L\frac{|C|}{|2A|} + \frac{3}{4} \geq 1$, i.e.,
$\frac{|C|}{|2A|} \geq \frac{1}{4L}.$

By Lemmas~\ref{lem:lowerBoundOnZ} and \ref{lem:YsquareIsSmall}
we have
\begin{align}\label{eq:6}
\Pr[Z >0] 
\geq \frac{(E[Z])^2}{E[Z^2]} 
\geq  \frac{(E[Z])^2}{E[Y^2]} \geq
\Omega\Bigl(\frac{1}{L^4}\Bigr)\;.
\end{align}

We want to count the number of
tuples $(\beta_1,\beta_2,\beta_3,\beta_4)$ such that 
$\beta_1+\beta_2+\beta_3+\beta_4 \in 2A$, and further 
$\beta_1+\beta_2 \in C$ and $\beta_3+\beta_4 \in C$.
By (\ref{eq:6}) there are at least
$\Omega(\frac{1}{L^4}) |A|^4$ tuples with 
$\beta_1+\beta_2+\beta_3+\beta_4 \in 2A$ and 
$|A(\beta_1+\beta_2)|, |A(\beta_1+\beta_2)| \geq \frac{|A|}{2K}$.
Of those, at most $\frac{|A|^4}{K^8}$ also satisfy 
$|A(\beta_1+\beta_2)| \geq \frac{L|A|}{K}$, and so 
there are still at least $\Omega(\frac{1}{L^4}) |A|^4$ tuples
with $\beta_1+\beta_2+\beta_3+\beta_4 \in 2A$, 
$\beta_1+\beta_2 \in C$, and $\beta_3+\beta_4 \in C$ (we can assume
that $L \leq K/2$, as otherwise the theorem is trivial).

This implies that there are at least $\Omega(\frac{1}{L^6})
K^2 |A|^2$ pairs  $(\gamma_1,\gamma_2) \in C \times C$ which
satisfy $\gamma_1 + \gamma_2 \in 2A$: each such pair can be
split into at most $\bigl(\frac{L|A|}{K}\bigr)^2$ four-tuples.

Thus, if we pick $\gamma_1$ and $\gamma_2$ uniformly in $C$
we have
\begin{align}
\Pr[\gamma_1 + \gamma_2 \in 2A] \geq \Omega\Bigl(\frac{1}{L^{6}}\Bigr)\;.
\end{align}

By the Balog-Szemer\'edi-Gowers Theorem (Theorem~\ref{thm:bsg}) 
applied on the set 
$A_{\mathrm{Thm\ref{thm:bsg}}}=C$ and $S_{\mathrm{Thm\ref{thm:bsg}}}=2A$
we get $C' \subseteq C$ such that $|C'| \geq |C|/\poly(L)$ and 
$|C'+C'| \leq \poly(L)|C|$.
By the Freiman-Ruzsa Theorem we get that $\sspan(C') \leq |C| 2^{\poly(L)}$.


Consider now the graph $G$ whose vertices are the elements
of $A$, and where we $a_1$ and $a_2$ is connected by an edge iff 
$a_1+a_2 \in C'$.
If two vertices $a_1$ and $a_2$ are in the same connected component 
of this graph, then $a_1+a_2$ is in $\sspan(C')$: $a_1+a_2$ equals
the sum of the elements in $C'$ which
correspond to the edges of a path from $a_1$ to $a_2$.

The graph has at least $|A|^2/\poly(L)$ edges, because each element
in $C'$ gives rise to at least $\frac{|A|}{2K}$ edges, and $|C'| \geq
\frac{K|A|}{\poly(L)}$.
Thus, the graph must have a connected component with at least $|A|/\poly(L)$ 
vertices (if the largest connected component takes 
fraction $\delta$ of the vertices, we get
at most $\delta^2 \cdot 1/\delta$ fraction of the edges).
Let $A' \subseteq A$ be the elements of such a component.
Because $2A' \subseteq \sspan(C')$ we also have $\sspan(2A') \subseteq
\sspan(\sspan(C')) = \sspan(C')$.
Furthermore, for every fixed $a' \in A'$ we  have
$\sspan(A') = \sspan(\{a'\} \cup 2A') \leq 2 \cdot |\sspan(2A')| \leq 
2\cdot |\sspan(C')| \leq |C|\cdot 2^{\poly(L)}$.
\end{proof}
\section{The Structured Case}

\begin{theorem}\label{thm:structured}
Let $A \subseteq \bbF_2^n$, $|A+A|= K|A|$.
Suppose that there is an element $a^*$, $L \in \bbN$,
and  $\epsilon > 10/\log(K)$
such that the set
\begin{align}
B = \{b \in A : |A({a^*+b}) | > |A|(1-K^{-\epsilon})\}
\end{align}
has size at least $K^{-L} |A|$.
Then, $B$ has span at most $2^{(4^{L/\epsilon})}|A+A|$.
\end{theorem}
\begin{proof}
Define
\begin{align}
S(\delta) := \{s \in A+A: |A(s)| > |A|(1-\delta)\}\;.
\end{align}
Clearly, $S(1) = |A+A|$. 
Furthermore, we see that $B+B \subseteq S(2K^{-\epsilon})$,
and also $S(\delta)+S(\delta) \subseteq S(2\delta)$.

Consider the sequence $2K^{-\epsilon}, 2\cdot2K^{-\epsilon}, 4\cdot2K^{-\epsilon},
\ldots, 2^{r}K^{-\epsilon}$ 
where $r$ is chosen such that the last term is $(\frac{1}{2},1]$.
We see that 
$r = \lceil\log_2(K^\epsilon/2)\rceil \geq \frac{\epsilon\log_2(K)}{2}$.

Since $K^{-L}|A| \leq |S(2K^{-\epsilon})|$ and
$|S(2^{r}K^{-\epsilon})| 
\leq |A+A| \leq K|A|$, there is some $j \in \{1,\ldots,r-1\}$ 
for which 
\begin{align}
|S(2^{j+1} K^{-L})| \leq K^{(L+1)/(r-1)}
|S(2^{j}K^{-L})|\;,
\end{align}
which implies that $S(2^j K^{-L})$ has doubling constant 
at most $K^{(L+1)/(r-1)}$.

Since $K^{(L+1)/(r-1)} \leq 2^{2L\log_2(K)/\epsilon\log_2(K)} = 4^{L/\epsilon}$,
by the Freiman-Ruzsa theorem, $S(2^{j}K^{-L})$ 
has span at most $2^{2\cdot 4^{L/\epsilon}}|S(2^j K^{-L})| \leq
2^{2\cdot 4^{L/\epsilon}}|A+A|$.
This in particular means that $2B$ has at most this span.
Since for any $b \in B$ we have that
$b+B \subseteq 2B$, we see that $\sspan(B) \leq 2 \sspan(2B)$, and 
so we get the result.
\end{proof}
Finally, we note that the theorem implies the statement 
given by item (2) in the abstract:
since $\Pr_{a_1,a_2 \in A}[|A(a_1+a_2)| \geq |A|(1-K^{-1/L})] \geq
K^{-L-1}$, there must be some $a_1$ for which 
$\Pr_{a_2\in A}[|A(a_1+a_2)| \geq |A|(1-K^{-1/L})] \geq K^{-L-1}$.
Applying Theorem~\ref{thm:structured} gives the claim, as long
as $L$ is smaller than some function of $K$, which (using
the Freiman-Ruzsa Theorem again) is clearly sufficient.

\section*{Acknowledgements}
I would like to thank Eli Ben-Sasson and Noga Ron-Zewi for inspiring
discussions and helpful pointers.
\bibliographystyle{alpha}
\bibliography{bibliography}

\end{document}